\documentclass{article}
\usepackage[english]{babel}

\usepackage{tikz}
\usetikzlibrary{quantikz2}

\usepackage{amssymb}
\usepackage{amsthm}
\usepackage{authblk}

\newtheorem{lemma}{Lemma}
\newtheorem{corollary}{Corollary}

\newcommand{\myequation}{Eq.}
\newcommand{\myfig}{Fig.}
\newcommand{\mylemma}{Lemma}

\title{Efficient Equivalent of\\ Shallow Quantum Hashing}

\author[1]{Ilnar Zinnatullin\thanks{\texttt{IlnGZinnatullin@kpfu.ru}}}
\author[1,2]{Alexander Vasiliev \thanks{\texttt{vav.kpfu@gmail.com}}}
\affil[1]{Kazan Federal University, Kazan, Russia}
\affil[2]{Zavoisky Physical-Technical Institute, FRC Kazan Scientific Center of RAS, Kazan, Russia}

\date{November, 2025}

\begin{document}

\maketitle

\begin{abstract}
Quantum hashing is a widely used technique in quantum computation that allows us to design space-efficient algorithms and protocols. Recently, Vasiliev has shown that the phase form of shallow quantum hashing can be implemented by a circuit of depth \(2\).

In this paper, we establish a connection between shallow quantum hashing and single-qubit quantum hashing for the amplitude form. For a shallow circuit, we propose a circuit of depth \(1\) that achieves the same collision resistance.
\end{abstract}

\section{Introduction}

Current quantum machines are NISQ (Noisy Intermediate-Scale Quantum) devices \cite{preskill2018}. They are affected by decoherence and are noisy, i.e., not fault-tolerant. Moreover, each quantum machine's architecture has its own set of elementary gates -- typically, one-qubit rotations and the two-qubit CNOT gate -- as well as limited qubit connectivity. Thus, we need time-efficient quantum algorithms that can be efficiently compiled (transpiled) and mapped to NISQ hardware so that they are executed before fragile quantum states are broken.

Quantum hashing was introduced in \cite{ablayev2013}.
Later, in \cite{vasiliev2016}, its generalized version based on the notion of \(\varepsilon\)-biased sets for finite abelian groups was proposed. A core part of quantum hashing and some of the other quantum techniques like quantum fingerprinting \cite{buhrman2001} is a uniformly controlled rotation (UCR) gate \cite{mottonen2004}. An \(n\)-qubit UCR gate uses \(n-1\) control qubits to select the rotation angle applied to the remaining target qubit, resulting in \(2^{n-1}\) distinct angles. In \cite{mottonen2004}, an efficient decomposition of a UCR into an alternating sequence of one-qubit rotations and CNOT gates is proposed. On current quantum hardware, UCR implementation is costly as it requires an exponential number of elementary gates with respect to the number of qubits, leading to exponential circuit depth. One of the ways to simplify the technique, reduce the circuit depth, and make it more suitable for the current hardware is to replace multi-qubit controlled rotations with two-qubit controlled rotations. K\={a}lis, in his master thesis \cite{kalis2018}, proposed a so-called shallow circuit for a quantum automaton with \(3\) qubits. Later, Ziiatdinov et al. generalized this approach for the quantum fingerprinting algorithm and presented a shallow implementation of quantum finite automata \cite{ziiatdinov2023}. Using this approach, Vasiliev proposed an extremely time-efficient algorithm for the phase form of shallow quantum hashing \cite{vasiliev2023}.

Quantum hashing has gone beyond theory to experiment: its first experimental implementation used single-photon states encoded in orbital angular momentum (OAM) \cite{turaykhanov2021}. The technique was subsequently generalized and implemented for high-dimensional qudits using the same technology \cite{akatev2022}. Note that the construction from \cite{vasiliev2023} coincides with the quantum hash function based on the single-photon states with OAM encoding \cite{turaykhanov2021}. The existence of a suitable parameter set that guaranties high collision resistance for the OAM-based quantum hashing is proved in \cite{vasiliev2023_1}.

In this work, we present a result similar to that obtained in \cite{vasiliev2023}. That is, for the amplitude form of the shallow quantum hashing, we reduce the circuit depth down to one. In this case, the circuit consists of only one layer of single-qubit rotations. This is advantageous on the current quantum hardware, as it requires only one-qubit rotations about the \(y\)-axis on the Bloch sphere.

The paper is organized in the following way. Section \(2\) reviews the necessary preliminaries. Section \(3\) is devoted to single-qubit quantum hashing and Section \(4\) to shallow quantum hashing. In Section \(5\), we establish a connection between these two versions of the technique and propose an efficient implementation of shallow quantum hashing. Finally, we summarize our results in Section \(6\).
 
\section{Preliminaries}

\paragraph{\(\varepsilon\)-biased sets}


The bias of a set \(B = \{b_{0}, b_{1}, \ldots, b_{d-1}\} \subseteq \mathbb{Z}_{q}\) with respect to \(x \in \mathbb{Z}_{q}\) is defined in the following way:
\[
\operatorname{bias}(B, x) = \frac{1}{|B|}\left|\sum_{b \in B} e^{i\frac{2\pi bx}{q}}\right|.
\]

The set \(B\) is called \(\varepsilon\)-biased if \(\max_{x \neq 0}\operatorname{bias}(B, x) \leq \varepsilon\). Let \linebreak \(B = \{b_{0}, b_{1}, \ldots, b_{d-1}\}\) and \(B' = \{0, (b_{1}-b_{0}), \ldots, (b_{d-1}-b_{0})\}\). Then, for any \(x \in \mathbb{Z}_{q}\), \(\operatorname{bias}(B, x) = \operatorname{bias}(B', x)\), in other words, \(B\) and \(B'\) are equivalent in terms of its bias \cite{akatev2022}. Thus, without loss of generality, we consider the first element \(b_{0}\) to be equal to \(0\) in the rest of the paper. Note that \cite{ablayev2023}
\[
\frac{1}{d} \left| \sum_{j=0}^{d-1} \cos\left( \frac{2\pi b_{j}x}{q} \right) \right| \leq \frac{1}{d} \left| \sum_{j=0}^{d-1} e^{i\frac{2\pi b_{j}x}{q}} \right| \leq \varepsilon
\]
for any \(x \in \mathbb{Z}_{q}\setminus\{0\}\). 
It has been shown \cite{noga1990} that \(\varepsilon\)-biased sets of size \linebreak\(d = O\left( \frac{\log q}{\varepsilon^{2}} \right)\) exist.

\paragraph{Quantum computation}

We use the Dirac notation. Mathematically, a qubit is represented by a normalized column vector (ket vector) \(\ket{\psi} = \alpha_{0}\ket{0} + \alpha_{1}\ket{1}\), \(\ket{\psi}\in \mathcal{H}^{2}\) with \(\alpha_{0}, \alpha_{1} \in \mathbb{C}\) and \(|\alpha_{0}|^{2} + |\alpha_{1}|^{2} = 1\). The state of \(n\) qubits is a ket vector from \((\mathcal{H}^2)^{\otimes n}\) of the form \(\sum_{j=0}^{2^{n}-1}\alpha_{j}\ket{j}\), where \(\ket{j}\) denotes the computational-basis state corresponding to the binary representation of \(j\), \linebreak\(\alpha_{j} \in \mathbb{C}\) for  \(0 \leq j \leq 2^{n-1}-1\) and \(\sum_{j=0}^{2^{n}-1}|\alpha_{j}|^{2} = 1\).

The quantum circuit model \cite{nielsen2010} is a quantum computation model, analogous to the classical circuit model. A quantum circuit is essentially a specification of a quantum algorithm. It consists of qubits (that are depicted as horizontal lines), quantum gates that implement unitary transformations, and measurements that extract classical information from the qubits. The time complexity of the algorithm is the circuit depth, the space complexity of the algorithm is the circuit width, i.~e. the number of qubits. For further reading on quantum circuits, we refer the reader to \cite{nielsen2010}.

We use the following elementary gates: the Hadamard gate \(H =
\begin{pmatrix}
\frac{1}{\sqrt{2}} & \frac{1}{\sqrt{2}}\\
\frac{1}{\sqrt{2}} & -\frac{1}{\sqrt{2}} 
\end{pmatrix}\) and rotation about the \(y\)-axis on the Bloch sphere \(R_{y}(\theta) =
\begin{pmatrix}
\cos(\frac{\theta}{2}) & -\sin(\frac{\theta}{2})\\
\sin(\frac{\theta}{2}) & \cos(\frac{\theta}{2}) 
\end{pmatrix}\). We also use controlled rotations about the \(y\)-axis: two-qubit controlled rotation (see \myfig~\ref{fig:rotation_gates}a) and  multi-qubit controlled rotation (see \myfig~\ref{fig:rotation_gates}b). By convention, a filled (black) dot denotes that the control qubit is in state \(\ket{1}\), whereas an open (white) dot indicates \ket{0} control. Thus, in controlled rotations, the rotation gate is applied iff all of the control qubits are in the states indicated by the control symbols.


\begin{figure}[h!]
    \centering
    \includegraphics[width=0.95\linewidth]{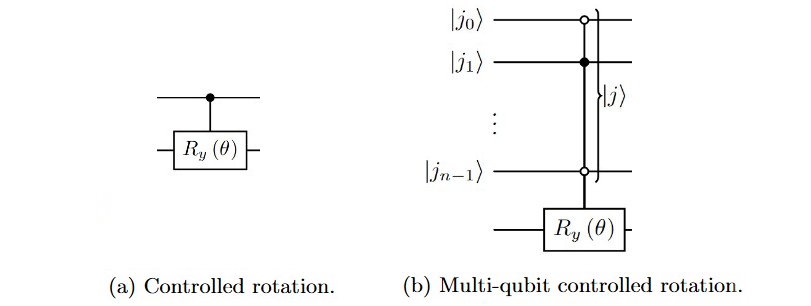}
    \caption{Controlled rotation gates.}
    \label{fig:rotation_gates}
\end{figure}

\paragraph{Quantum Hashing}

Let \(\mathbb{X}\) be a finite set with \(|\mathbb{X}| = N\). Let \(\psi : \mathbb{X} \rightarrow (\mathcal{H}^{2})^{\otimes n}\) be a classical-quantum function that maps an element of \(\mathbb{X}\) to the quantum state of \(n\) qubits. In general, for quantum hashing, \(N \gg n\).

In the quantum world, there are no collisions in the classical sense as different inputs are mapped to different outputs. However, a procedure for comparing of quantum states (equality test) is probabilistic, and we need to distinguish distinct hashes with high probability. The probabilistic nature of quantum computation may lead us to measurement errors, so the equality test may fail with non-zero probability. Since non-orthogonal states cannot be perfectly distinguished, we require them to be ``nearly orthogonal''.

We formalize ``nearly orthogonal'' states as follows. For \(\varepsilon \in [0, 1)\), states \(\ket{\psi}\) and \(\ket{\psi'}\) are \(\varepsilon\)-orthogonal if \(|\braket{\psi}{\psi'}| \leq \varepsilon\). The function \(\psi\) is \(\varepsilon\)-resistant (to collisions) if for any pair of distinct \(x_{1}, x_{2} \in \mathbb{X}\) their quantum images are \(\varepsilon\)-orthogonal, i.e., \(|\braket{\psi(x_{1})}{\psi(x_{2})}| \leq \varepsilon\). More details on quantum hashing can be found in \cite{bookablayev2015,bookablayev2023}.

Throughout the paper, we set \(\mathbb{X} = \mathbb{Z}_{q}\).

\section{Single-Qubit Quantum Hashing}

Let \(S = \{s_{0}, s_{1}, \ldots, s_{n-1}\} \subseteq \mathbb{Z}_{q}\) be a set of parameters. We define a classical-quantum function
\[
\psi_{S} : \mathbb{Z}_{q} \rightarrow (\mathcal{H}^{2})^{\otimes n}
\]
that maps elements of \(\mathbb{Z}_{q}\) to quantum states. That is, for an input \(x \in \mathbb{Z}_q\), its amplitude-form quantum hash \(\ket{\psi_{S}(x)}\) is composed of \(n\) single qubits:
\[
\ket{\psi_{S}(x)} = \ket{\psi_{S}^{0}(x)} \otimes \ket{\psi_{S}^{1}(x)} \otimes \ldots \otimes \ket{\psi_{S}^{n-1}(x)},
\]
where
\[
\ket{\psi_{S}^{j}(x)} = \cos\left(\frac{\pi s_{j}x}{q}\right)\ket{0} + \sin\left(\frac{\pi s_{j}x}{q}\right)\ket{1} = R_{y}\left(\frac{2\pi s_{j}x}{q}\right)\ket{0}
\]
for \(0 \leq j \leq n-1\).

\begin{lemma}\label{lemma:scalar_product_for_single-qubit_qh}
For the function \(\psi_{S}\) and any \(x_{1}, x_{2} \in \mathbb{Z}_{q}\), we have
\[
\braket{\psi_{S}(x_{1})}{\psi_{S}(x_{2})} = \prod_{j=0}^{n-1}\cos\left( \frac{\pi s_{j}(x_{1}-x_{2})}{q} \right).
\]
\end{lemma}

\begin{proof}
\begin{equation*}
\begin{split}
\braket{\psi_{S}(x_{1})}{\psi_{S}(x_{2})} &= \prod_{j=0}^{n-1} \braket{\psi_{S}^{j}(x_{1})}{\psi_{S}^{j}(x_{2})}\\
                                  &= \prod_{j=0}^{n-1} \left( \cos\left(\frac{\pi s_{j}x_{1}}{q}\right)\cos\left(\frac{\pi s_{j}x_{2}}{q}\right)+\sin\left(\frac{\pi s_{j}x_{1}}{q}\right)\sin\left(\frac{\pi s_{j}x_{2}}{q} \right) \right)\\
                                  &= \prod_{j=0}^{n-1} \cos\left(\frac{\pi s_{j}(x_{1}-x_{2})}{q}\right).
\end{split}
\end{equation*}

\end{proof}

\begin{lemma}
There exists \(S \subseteq \mathbb{Z}_{q}\) with \(|S| = O(\log q)\) such that the function
\[
    \psi_{S} : \mathbb{Z}_{q} \rightarrow (\mathcal{H}^{2})^{\otimes n},
\]
\[
    \psi_{S} : x \mapsto \ket{\psi_{S}^{0}(x)} \otimes \ldots \otimes \ket{\psi_{S}^{n-1}(x)},\, \ket{\psi_{S}^{j}} = R_{y}\left(\frac{2\pi s_{j}x}{q}\right)\ket{0}
\]
is collision-resistant.
\end{lemma}
\begin{proof}
    Using \mylemma~\ref{lemma:scalar_product_for_single-qubit_qh}, we deduce that for \(x_{1}, x_{2} \in \mathbb{Z}_{q}\)
    \[
        \braket{\psi_{S}(x_{1})}{\psi_{S}(x_{2})} = \prod_{j=0}^{n-1} \cos\left(\frac{\pi s_{j}(x_{1}-x_{2})}{q}\right).
    \]
    For the rest of the proof, see \cite[Lemma 2]{vasiliev2023}.
\end{proof}

\section{Shallow Quantum Hashing}

Let \(B = \{b_{0}, b_{1}, \ldots, b_{d-1}\}\) be an \(\varepsilon\)-biased set. We define a classical-quantum function
\[
\psi_{B} : \mathbb{Z}_{q} \rightarrow (\mathcal{H}^{2})^{\otimes \log d}
\]
that maps elements of \(\mathbb{Z}_{q}\) to quantum states.  For \(x \in \mathbb{Z}_{q}\), the function \(\psi_{B}\) defines the amplitude form of quantum hashing as follows \cite{ablayev2013}:
\begin{equation}\label{eq:qh_standard_form}
    \begin{split}
        \ket{\psi_{B}(x)} &= \frac{1}{\sqrt{d}} \sum_{j=0}^{d-1}\ket{j} \left(\cos\left(\frac{2\pi b_{j}x}{q}\right)\ket{0} + \sin\left(\frac{2\pi b_{j}x}{q}\right)\ket{1}\right)\\
        &= \frac{1}{\sqrt{d}} \sum_{j=0}^{d-1}\ket{j} \otimes R_{y}\left(\theta_{j}\right)\ket{0}, 
    \end{split}
\end{equation}
where \(\theta_{j} = \frac{4\pi b_{j}x}{q}\), \(0 \leq j \leq d - 1\).

To construct a quantum state defined by \myequation~(\ref{eq:qh_standard_form}), a circuit presented in \myfig~\ref{fig:qh_standard_form} is used.


\begin{figure}[h!]
    \centering
    \includegraphics[width=0.70\linewidth]{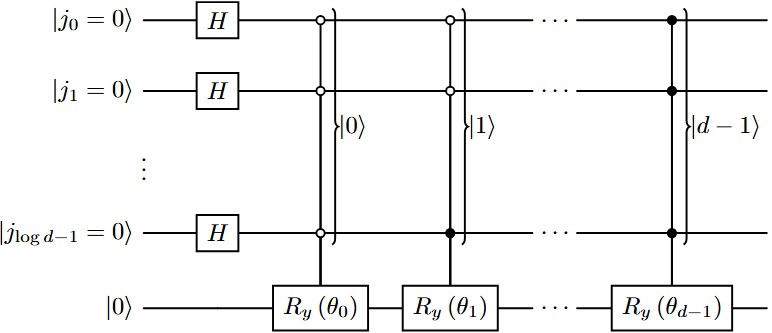}
    \caption{A circuit for amplitude form of quantum hashing.}
    \label{fig:qh_standard_form}
\end{figure}

Let \(n = \log d\) and \(S = \{s_{0}, s_{1}, \ldots, s_{n-1}\} \subseteq \mathbb{Z}_{q}\). We consider a special case of the set \(B\) when its elements are linear combinations of elements of \(S\): 
\begin{equation}\label{eq:b_element}
b_{j} = \sum_{k=0}^{n-1} j_{k}s_{k},\, j = j_{0}j_{1}\ldots j_{n-1},\quad 0 \leq j \leq d-1.
\end{equation}

\begin{lemma}\label{lemma:shallow_circuit}

The circuits illustrated in \myfig~\ref{fig:ucr} and \myfig~\ref{fig:shallow_ucr} are equivalent for angles \(\theta_{j}\) and \(\gamma_{k}\), with \(\theta_{j}\) given by
\begin{equation}\label{eq:theta_decomposition}
    \theta_{j} = \gamma + \sum_{k=0}^{n-1} j_{k}\gamma_{k},
\end{equation}
where \(j = j_{0}j_{1}\ldots j_{n-1}\) with \(j_{k} \in \{0, 1\}\) for \(0 \leq j \leq 2^{n}-1\) and \(0 \leq k \leq n-1\).


\begin{figure}[h!]
    \centering
    \includegraphics[width=0.65\linewidth]{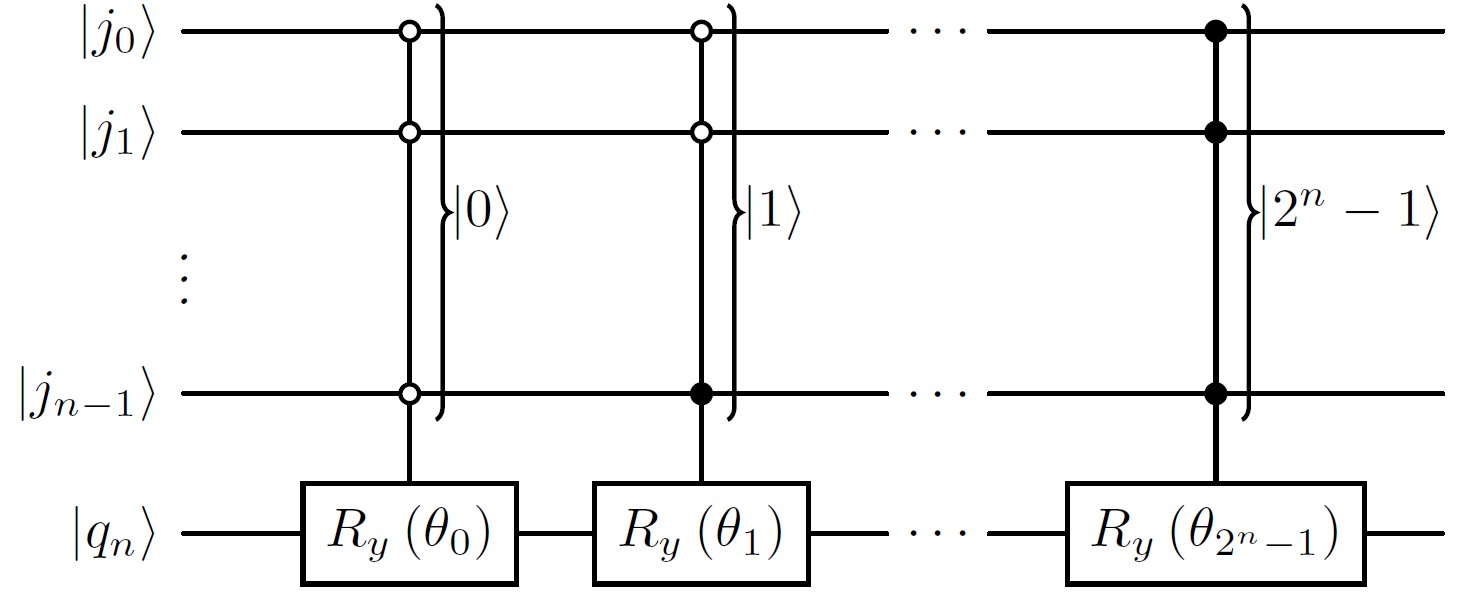}
    \caption{Uniformly controlled rotation.}
    \label{fig:ucr}
\end{figure}

\begin{figure}[h!]
    \centering
    \includegraphics[width=0.70\linewidth]{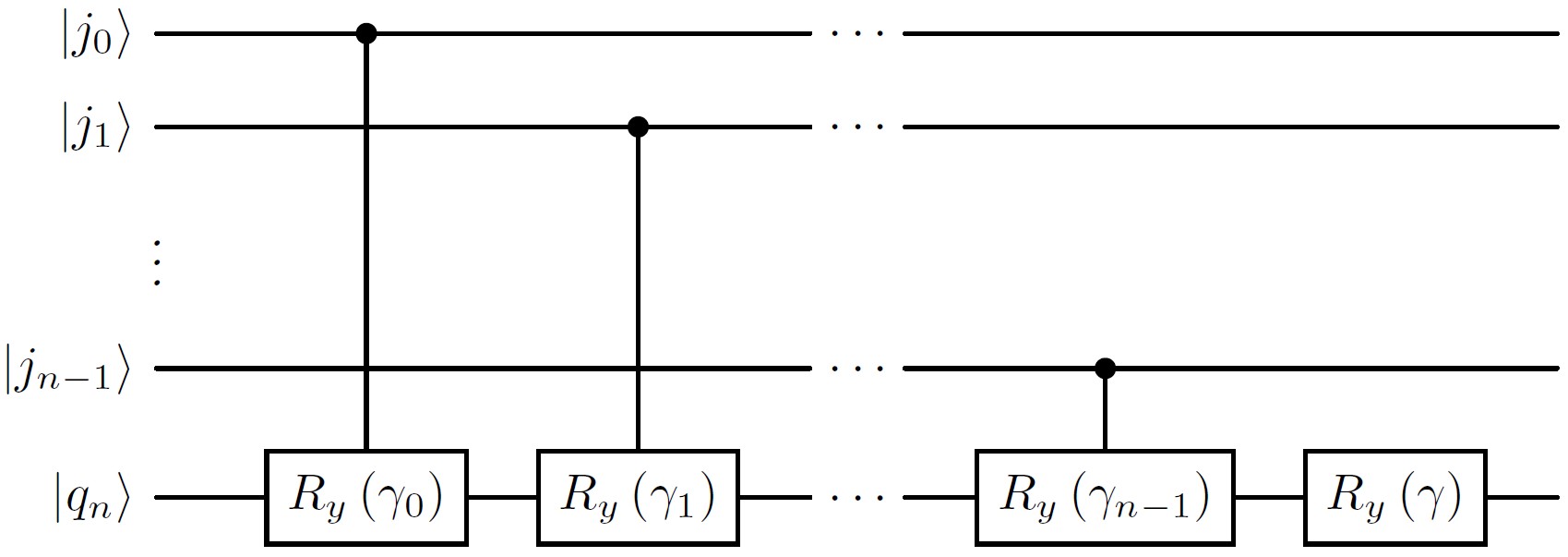}
    \caption{Shallow uniformly controlled rotation decomposition.}
    \label{fig:shallow_ucr}
\end{figure}

\end{lemma}
\begin{proof}
Let \(\ket{\psi} = \ket{j_{0}j_{1}\ldots j_{n-1}} \otimes \ket{q_{n}}\) be a fixed computational-basis quantum state with \(q_n \in \{0,1\}\). Consider the circuit in \myfig~\ref{fig:shallow_ucr}. The gate
\[
R_{y}(\gamma) \prod_{k:j_{k}=1} R_{y}(\gamma_{k}) = R_{y} \left( \gamma + \sum_{k=0}^{n-1} j_{k}\gamma_{k} \right) = R_{y}(\theta_{j})
\]
is applied to the target qubit \(\ket{q_{n}}\) as rotations about the same axis are additive.

In the circuit shown in \myfig~\ref{fig:ucr}, by definition of the UCR gate \cite{mottonen2004}, the operation \(R_{y}(\theta_{j})\) is applied to the target qubit \(\ket{q_{n}}\). Since both circuits act identically on all computational-basis states, they are equivalent.
\end{proof}

Using \myequation~(\ref{eq:b_element}) and \mylemma~\ref{lemma:shallow_circuit}, we note that a simplified circuit shown in \myfig~\ref{fig:qh_shallow_form} can be used with rotation angles \(\gamma_{j} = \frac{4\pi s_{j}x}{q}\). Note that gate \(R_{y}(\gamma)\) is omitted as in \myequation~(\ref{eq:theta_decomposition}) \(\gamma = 0\) in our case.


\begin{figure}[h!]
    \centering
    \includegraphics[width=0.65\linewidth]{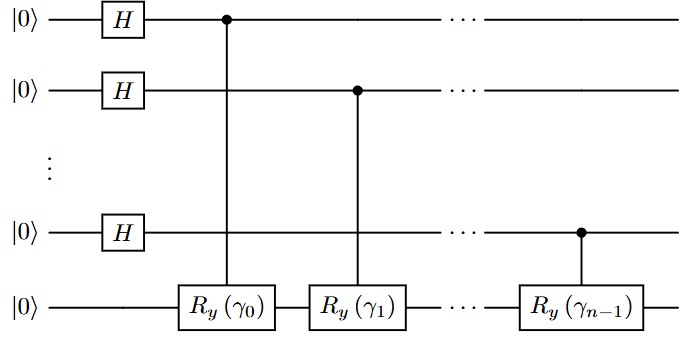}
    \caption{A circuit for shallow quantum hashing.}
    \label{fig:qh_shallow_form}
\end{figure}

Let us define for \(S\) and \(0 \leq j \leq 2^{n}-1\) 
\[f(S, j) = j_{0}s_{0} + j_{1}s_{1} + \ldots + j_{n-1}s_{n-1} = \sum_{k=0}^{n-1} j_{k}s_{k},\] 
where \(j_{0}j_{1} \ldots j_{n-1}\) is the binary representation of \(j\). Then, \(b_{j} = f(S, j)\) and \myequation~(\ref{eq:qh_standard_form}) can be rewritten as follows:

\begin{equation*}
    \begin{split}
        \ket{\psi_{B}(x)} = \ket{\widetilde{\psi}_{S}(x)} &= \frac{1}{\sqrt{2^{n}}} \sum_{j=0}^{2^{n}-1}\ket{j} \left(\cos\left(\frac{2\pi f(S, j)x}{q}\right)\ket{0} + \sin\left(\frac{2\pi f(S, j)x}{q}\right)\ket{1}\right)\\
        &= \frac{1}{\sqrt{2^{n}}} \sum_{j=0}^{2^{n}-1}\ket{j} \otimes R_{y}\left(\frac{4\pi f(S, j)x}{q}\right)\ket{0}.
    \end{split}
\end{equation*}

\begin{lemma}\label{lemma:scalar_product_for_shallow_qh}
For the function \(\widetilde{\psi}_{S}\) and any pair of \(x_{1}, x_{2} \in \mathbb{Z}_{q}\), we have
\[
\braket{\widetilde{\psi}_{S}(x_{1})}{\widetilde{\psi}_{S}(x_{2})} = \cos\left(\frac{\pi (x_{1}-x_{2})}{q} \sum_{l=0}^{n-1} s_{l}\right)\prod_{k=0}^{n-1} \cos\left(\frac{\pi s_{k} (x_{1}-x_{2})}{q}\right).
\]
\end{lemma}

\begin{proof}
By \(S_{k}\), we define
\[
S_{k} = j_{0}s_{0} + \ldots + j_{k}s_{k}, 
\]
for \(0 \leq k \leq n-1\).

We use the following equation:
\[
\cos(S_{k-1}) + \cos(S_{k-1} + s_{k}) = 2\cos\left(S_{k-1}+\frac{s_{k}}{2}\right)\cos\left(\frac{s_{k}}{2}\right)
\]
for \(1 \leq k \leq n-1\).

Let \(x = x_{1} - x_{2}\). Let us consider the inner product \(\braket{\widetilde{\psi}_{S}(x_{1})}{\widetilde{\psi}_{S}(x_{2})}\) of hashes corresponding to inputs \(x_{1}\) and \(x_{2}\):

\[
\begin{split}
&\frac{1}{2^{n}} \sum_{j=0}^{2^{n}-1} \cos\left(\frac{2\pi f(S, j)x_{1}}{q}\right)\cos\left(\frac{2\pi f(S, j)x_{2}}{q}\right)+ \sin\left(\frac{2\pi f(S, j)x_{1}}{q}\right)\sin\left(\frac{2\pi f(S, j)x_{2}}{q}\right)\\
&= \frac{1}{2^{n}} \sum_{j=0}^{2^{n}-1} \cos\left(\frac{2\pi f(S, j)(x_{1}-x_{2})}{q}\right) = \frac{1}{2^{n}} \sum_{j=0}^{2^{n}-1} \cos\left(\frac{2\pi f(S, j)x}{q}\right)\\
&= \frac{1}{2^{n}} \sum_{j_{0}=0}^{1}\sum_{j_{1}=0}^{1}\ldots \sum_{j_{n-1}=0}^{1} \cos\left(\frac{2\pi x}{q}(j_{0}s_{0}+j_{1}s_{1}+\ldots+j_{n-1}s_{n-1})\right)\\
&= \frac{1}{2^{n}} \sum_{j_{0}=0}^{1}\ldots \sum_{j_{n-2}=0}^{1} \cos\left(\frac{2\pi x}{q} S_{n-2}\right) + \cos\left(\frac{2\pi x}{q}(S_{n-2}+s_{n-1})\right)\\
&= \frac{1}{2^{n}} \sum_{j_{0}=0}^{1}\ldots \sum_{j_{n-2}=0}^{1} 2\cos\left(\frac{2\pi x}{q}
\left(S_{n-2}+\frac{s_{n-1}}{2}\right)\right)\cos\left(\frac{\pi s_{n-1}x}{q}\right)\\
&= \frac{1}{2^{n-1}}\cos\left(\frac{\pi s_{n-1}x}{q}\right) \sum_{j_{0}=0}^{1}\ldots \sum_{j_{n-2}=0}^{1} \cos\left(\frac{2\pi x}{q}\left(S_{n-2}+\frac{s_{n-1}}{2}\right)\right)\\
&= \frac{1}{2^{n-1}}\cos\left(\frac{\pi s_{n-1}x}{q}\right) \sum_{j_{1}=0}^{1}\ldots \sum_{j_{n-3}=0}^{1} \cos\left(\frac{2\pi x}{q}\left(S_{n-3}+\frac{s_{n-1}}{2}\right)\right)+\\
&\cos\left(\frac{2\pi x}{q}\left(S_{n-3}+\frac{s_{n-1}}{2}+s_{n-2}\right)\right)\\
                                  &= \frac{1}{2^{n-1}}\cos\left(\frac{\pi s_{n-1}x}{q}\right) \sum_{j_{1}=0}^{1}\ldots \sum_{j_{n-3}=0}^{1} 2\cos\left(\frac{2\pi x}{q}\left(S_{n-3}+\frac{s_{n-1}}{2}+\frac{s_{n-2}}{2}\right)\right)\cos\left(\frac{\pi s_{n-2}x}{q}\right)\\
                                  &= \frac{1}{2^{n-2}}\cos\left(\frac{\pi s_{n-1}x}{q}\right)\cos\left(\frac{\pi s_{n-2}x}{q}\right) \sum_{j_{0}=0}^{1}\ldots \sum_{j_{n-3}=0}^{1} \cos\left(\frac{2\pi x}{q}\left(S_{n-3}+\frac{s_{n-1}}{2}+\frac{s_{n-2}}{2}\right)\right)\\
                                  &= \cos\left(\frac{\pi x}{q} \sum_{l=0}^{n-1} s_{l}\right)\prod_{k=0}^{n-1} \cos\left(\frac{\pi s_{k} x}{q}\right).
\end{split}
\]
\end{proof}

\section{Efficient Algorithm for Shallow Quantum Hashing}

In this section, we establish a connection between shallow quantum hashing and quantum hashing via single qubits.

\begin{lemma}
Let \(S = \{s_{0}, \ldots, s_{n-1}\} \subseteq \mathbb{Z}_{q}\) be a set of parameters. The following quantum hash functions
\[
\ket{\widetilde{\psi}_{S}(x)} = \frac{1}{\sqrt{2^{n}}} \sum_{j=0}^{2^{n}-1}\ket{j} \otimes R_{y}\left(\frac{4\pi f(S, j)x}{q}\right)\ket{0}
\]
and
\[
\ket{\psi_{S}(x)} = \left(\bigotimes_{k=0}^{n-1} R_{y}\left(\frac{2\pi s_{j}x}{q}\right)\ket{0}\right) \otimes R_{y}\left(\frac{2\pi x}{q}\sum_{l=0}^{n-1}s_{j}\right)\ket{0}
\]
are equivalent in the sense of collision resistance, i.e. have the same collision resistance.
\end{lemma}

\begin{proof}
    Let \(x_{1}, x_{2} \in \mathbb{Z}_{q}\). By \mylemma~\ref{lemma:scalar_product_for_shallow_qh}, for the function \(\widetilde{\psi}_{S}\), we have
    \begin{equation}\label{eq:collision_resistance_shallow_qh}
        \braket{\widetilde{\psi}_{S}(x_{1})}{\widetilde{\psi}_{S}(x_{2})} = \cos\left(\frac{\pi (x_{1}-x_{2})}{q} \sum_{l=0}^{n-1} s_{l}\right)\prod_{k=0}^{n-1} \cos\left(\frac{\pi s_{k} (x_{1}-x_{2})}{q}\right).
    \end{equation}
    By \mylemma~\ref{lemma:scalar_product_for_single-qubit_qh}, for the function \(\psi_{S}\), we have
    \begin{equation}\label{eq:collision_resistance_single-qubit_qh}
        \braket{\psi_{S}(x_{1})}{\psi_{S}(x_{2})} = \prod_{k=0}^{n-1} \cos\left(\frac{\pi s_{k} (x_{1}-x_{2})}{q}\right)\cos\left(\frac{\pi (x_{1}-x_{2})}{q} \sum_{l=0}^{n-1} s_{l}\right).
    \end{equation}
    So, it is easy to see that the right parts of the Eqs.~(\ref{eq:collision_resistance_shallow_qh}) and (\ref{eq:collision_resistance_single-qubit_qh}) coincide.
\end{proof}

\begin{corollary}
The circuits in \myfig~\ref{fig:qh_shallow_form} and \myfig~\ref{fig:single-qubit_qh} are equivalent in the sense that they provide the same collision resistance. Note that \(\gamma_{k} = \frac{4\pi s_{k}x}{q}\), \(\gamma' = \frac{2\pi sx}{q}\), \(s = \sum_{j=0}^{n-1} s_{j}, x = x_{1} - x_{2}\).

\begin{figure}[h!]
    \centering
    \includegraphics[width=0.30\linewidth]{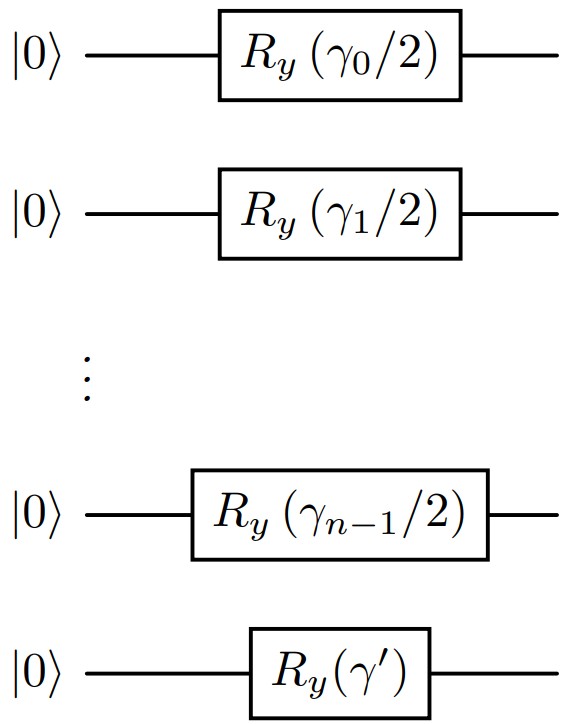}
    \caption{A circuit for single-qubit quantum hashing.}
    \label{fig:single-qubit_qh}
\end{figure}


\end{corollary}

Finally, the amplitude form of shallow quantum hashing can be implemented with a circuit for single-qubit quantum hashing.

\section{Conclusion}
We have analyzed the algorithm for quantum hashing and have established a connection between the amplitude form of the shallow quantum hashing and the amplitude form of single-qubit quantum hashing. In this paper, we propose an extremely time-efficient quantum hash function equivalent to the amplitude form of shallow quantum hashing in terms of collision resistance and the number of qubits. The new version uses no entanglement, which makes it even more practical. 


\bibliographystyle{unsrt}
\bibliography{lib.bib}

\end{document}